\documentclass[orivec]{llncs}

\usepackage{amsfonts}
\usepackage[hide]{ed}
\usepackage[all]{xy}
\usepackage{mathtools}
\usepackage{wrapfig}

\usepackage{misc}

\usepackage{lstsemantic}
\usepackage{paralist}
\usepackage{local}

\usepackage[bookmarks,linkcolor=red,citecolor=blue,colorlinks,breaklinks,bookmarksopen,bookmarksnumbered]{hyperref}

\title{Canonical Selection of Colimits}
\author{
Till Mossakowski\inst{1},
Florian Rabe\inst{2} and
Mihai Codescu\inst{3}
}
\institute{Otto-von-Guericke-University of Magdeburg, Germany\and
Jacobs University Bremen, Germany,\and
Free University of Bozen-Bolzano, Italy}

\begin{document}

\maketitle

\begin{abstract}
Colimits are a powerful tool for the combination of objects in a
category.  In the context of modeling and specification, they are used
in the institution-independent semantics (1) of instantiations of parameterised
specifications (e.g.\ in the
specification language CASL), and (2) of combinations of
networks of specifications (in the
OMG standardised language DOL).

The problem of using colimits as the semantics of certain language
constructs is that they are defined only up to isomorphism. However,
the semantics of a complex specification in these languages is given
by a signature and a class of models over that signature -- not by an
isomorphism class of signatures. This is particularly relevant when a
specification with colimit semantics is further
translated or refined. The user needs to know the symbols of a
signature for writing a correct 
refinement.

Therefore, we study how to usefully choose one representative of the
isomorphism class of all colimits of a given diagram. We develop
criteria that colimit selections should meet.  We work over arbitrary
inclusive categories, but start the study how the criteria can be
met with $\Set$-like categories, which are often used as signature
categories for institutions.
\end{abstract}

\section{Introduction}
\begin{quotation}
``Given a species of structure, say widgets, then the result
of interconnecting a system of widgets to form a super-widget
corresponds to taking the {\em colimit} of the diagram
of widgets in which the morphisms show how they are interconnected.''
\cite{CategoricalManifesto}
\end{quotation}

\paragraph{Motivation}
The notion of colimit provides a natural way to abstract the idea that
some objects of interest, which can be e.g.~logical theories, software
specifications or semiotic systems, are combined while taking into
account the way they are related. Specification languages whose
semantics involves colimits are CASL \cite{CASL-RM} (for instantiations of
parameterised specifications) and its extension DOL (see \cite{MossakowskiEtAl13d,DOL-OMG} and \url{http://dol-omg.org}) (for combination
of networks of specifications). Specware  \cite{journals/ase/WilliamsonHB01}
provides a tool computing colimits of specifications that has been
successfully used in industrial applications; \cite{DBLP:conf/birthday/Smith06}
makes a strong case for the use of colimits in formal software
development. The Heterogeneous Tool
Set (HETS, \cite{hets}) also supports the computation of colimits,
covering even the heterogeneous case \cite{weakcol}.
Colimits have been used for ontology alignment \cite{zimmermann-fois}
and database integration \cite{DBLP:journals/corr/SchultzSVW16}.
Recently,
colimits have provided the base mechanism for concept creation by
blending existing concepts \cite{KutzEtAl14b}. Moreover, colimits
provide the basis for a good behaviour of parameterisation in a
specification language \cite{Diaconescu1993}.\footnote
{When we use the term \emph{specification}, our theory applies
equally to \emph{ontologies} and \emph{models}, provided these
have a formal semantics as theories of some institution.}

The problem that arises naturally when using colimits is that they are
not unique, but only unique up to isomorphism. By contrast, the
semantics of a specification involves a specific signature, which must
be selected from this isomorphism class.  Also, any implementation of
colimit computation in a tool must make an according choice of how the
colimiting object actually looks, in particular when it comes to the
names of its symbols.
Otherwise, users have no control over the well-formedness of further specifications built from the colimit: Referring to symbols of the colimit is only possible with knowledge about the actual symbol names appearing in the colimit.
To be useful in practice, it is desirable that such a choice appears natural to the user.
For example, the names of original symbols should be preserved whenever possible.

\paragraph{Contribution}
Our contribution is two-fold.
Firstly, we develop a suite of properties that can be used to evaluate and classify different colimit selections.
All of these are motivated by the desire that parameterisation and combination of networks enjoy good properties.
We show that these properties, although all desirable, cannot be realised at once.
Secondly, we give solutions for systematically selecting colimits in various signature categories that provide good trade-offs between these conflicting properties.

\paragraph{Related work}
The semantics of CASL~\cite{BaumeisterEtAl04,Mossakowski00c} provides
some method for the computation of specific pushouts. However, the chosen
institutional framework (institutions with a lot of extra
infrastructure) is rather complicated, while we use the much more
natural framework of inclusive categories. Moreover, desirable properties
of pushouts are only discussed casually.
Rabe~\cite{rabe:howto:14}
discusses three desirable properties of selected pushouts and
conjectures that they are not reconcilable.  We shed light on this
conjecture and provide a total selection of pushouts, while
\cite{rabe:howto:14} only provides a partial selection.
In the context of Specware, colimits are computed as equivalence classes \cite{DBLP:conf/birthday/Smith06}.
The systematic investigation of selected colimits (i.e.\ beyond pushouts) is new to our knowledge.

\paragraph*{Overview}
In Section~\ref{sec:prelim}, we recall some preliminaries as well as language constructs from CASL and DOL that involve colimits in their semantics.
Then we develop criteria for elegant colimit selection in Section~\ref{sec:properties}.
In Section~\ref{sec:colimits}, we give colimit selections for various categories.
We will also see that not every category admits a selection that satisfies all desirable properties.
Therefore, we pursue a second goal in Section~\ref{sec:categories}, namely to find useful categories for which we can give particularly elegant selections.

Some proofs have been omitted.
These were available during review and will appear in an extended version of this paper.

\section{Preliminaries}\label{sec:prelim}
\subsection{Categories with Symbols}

The large variety of logical languages in use can be captured at an
abstract level using the concept of \emph{signature categories}.
The objects of such a category are signatures which introduce syntax for the domain of interest, and the signature morphisms capture relations between signatures such as changes of notation, extensions, or translations.
For example, signature categories feature heavily in the framework of \emph{institutions} \cite{GoguenBurstall92}, where they are the starting point for abstractly capturing the semantics of logical systems and developing results independently of the specific features of a logical system.

In institutions and similar frameworks, the signature category is abstract, i.e., it is an arbitrary category.
In practice, some properties of signature categories have emerged that are
 satisfied by the over-whelming majority of logical systems, and that are
 very helpful for establishing generic results.

For colimits, two properties are particularly important:

\begin{definition}[\cite{CazuanescuR1997}]  
An \defsty{inclusive category} consists of a category $\C$ with a broad subcategory\footnote{That is, with the same objects as $\C$.} that is a partially ordered class.

The morphisms of the broad subcategory are called \defsty{inclusions}, and we write $A\hra B$ if there is an inclusion from $A$ to $B$.
We denote the (quasi-)category of inclusive categories and inclusion
preserving functors by $\ICat$.
\end{definition}
In particular, $\Set$ is an inclusive category via the standard inclusions $A\hra B$ iff $A\subseteq B$.
Arbitrary categories can be recovered by using the identity relation as the partial order.


\begin{definition}
  An (inclusive) \emph{category with symbols} consists of an (inclusive) category $\C$ and an (inclusion-preserving) functor $|\_|:\C\to\Set$
    We call $|A|$ the set of symbols of $A$.
\end{definition}
In particular, $\Set$ is an inclusive category with symbols via $|A|=A$.

The intuition behind these definitions is that very often signatures can be seen as sets of named declarations.
Then the subset relation defines the inclusion relation, and the names of the declarations define the set of symbols.

Signature categories are usually such that signatures that differ only in the choice of names are isomorphic.
Then a key difficulty about colimits lies in selecting the set of names to be used in the colimit.\footnote{For inclusive categories for which the symbol functor uniquely lifts colimits, solving colimit selection for $\Set$ already
  suffices. However, the inclusive categories studied in Sect.~\ref{sec:split} typically
  do not enjoy this property.}


\ednote{Reviewer suggests: Consider inclusive categories with objects for which the symbol functor lifts uniquely the colimits of certain diagrams (e.g. name-clash-free diagrams). TM: see previous footnote.}

\subsection{Specification Operators with Colimit Semantics}\label{sec:constructs}

The power of the abstraction provided by institutions and related systems is best
illustrated by the fact that languages like CASL and DOL provide
syntax and semantics of specifications \emph{in an arbitrary institution}.
This is done by defining operators on specifications and morphisms.

A \emph{basic specification} consists of a signature and a set of sentences\footnote{It is straightforward but not essential here to make the notion of sentence precise.}---called the \emph{axioms}---over it.
A kernel language of specification operators has been introduced in \cite{SpecInst}.
It includes union, renaming and hiding.
CASL and DOL provide many further constructs.

The semantics of many of these operators can be defined as the colimit of a certain diagram.
Therefore, such operators are often defined only up to isomorphism.
In the sequel we recall important examples from CASL (parameterisation) and DOL (combination of networks).


\paragraph{Parametrisation}\label{sec:param}
Many specification languages, including CASL, allow specifications to
be generic. A generic specification consists of a (formal) parameter
specification $P$ and a body specification $B$ extending the formal
parameter, i.e.\ $P\hra B$.
We make $P$ explicit by writing $B[P]$.

A typical example is the specification $\mathit{List}[Elem]$ for lists parametrised by the specification $\mathit{Elem}$ which declares a sort $elem$.


Given an actual parameter specification $A$ and a specification morphism
$\sigma:P\to A$, we write the instantiation of $B[P]$ with $A$ via $\sigma$ as $B[A\textbf{ fit }\sigma]$.
Its semantics is given by the pushout on the left below:
 \[\xymatrix{
    P\, \ar@{^{(}->}[r] \ar[d]^{\sigma} & B \ar[d]\\
    A\, \ar@{^{(}->}[r] & B[A\textbf{ fit }\sigma]
  }\begin{array}{l}\\ \\ \text{ e.g. }\end{array}\xymatrix{
    \mathit{Elem}\, \ar@{^{(}->}[r] \ar[d]^{\sigma} & \mathit{List[Elem]} \ar[d]\\
    \mathit{Nat}\, \ar@{^{(}->}[r] & \mathit{List}[\mathit{Nat}\textbf{ fit } elem\mapsto nat]
  }\]
The right hand side above gives a typical example where the specification $\mathit{Nat}$ declares a sort $nat$ that is used to instantiate the sort $elem$.
Note that the above is also an example of how a pushout of an inclusion can often be selected as an inclusion again. We will get back to that in Def.~\ref{def:pushoutinclusion}.

A natural requirement is that the instantiated body $B[A\textbf{ fit }\sigma]$ extends the actual parameter $A$ in much the same way as the body $B$ extends the formal parameter $P$.
For example, a sort $\mathit{list}$ introduced in the specification $\mathit{List}$ should be kept (and not renamed) within the instantiation 
$\mathit{List}[\mathit{Nat}\textbf{ fit }elem \mapsto nat]$.
Technically, this means that the semantics should not be an \emph{arbitrary} colimit.
Similarly, the user would expect that any symbols declared in the body should appear verbatim in the instantiated body, unless they have been renamed by $\sigma$.


\paragraph{Networks of Specifications}\label{sec:networks}
In DOL, a network of specifications (called distributed specification
in \cite{MossakowskiTarlecki09}) is a graph. Its nodes are labelled
with pairs $(O, SP)$ where $SP$ is a specification and $O$ its
name. The edges are theory morphisms $(O_1,SP_1)\stackrel{\sigma}{\to}
(O_2, SP_2)$, either induced by the import structure of the
specifications, or by refinements.
    
\begin{wrapfigure}{r}{2.5cm}
\vspace{-3em}
\begin{lstlisting}
network N =
 $N_1, \ldots, N_m,$
 $O_1, \ldots, O_n,$
 $M_1, \ldots, M_p $
\end{lstlisting}
\vspace{-3em}
\end{wrapfigure}

A network is specified by giving a list of specifications $O_i$, morphisms $M_i$ between them
and sub-networks $N_i$, with the intuition that the graph of the network is the union of the graphs of all its elements.

Now the operator \textbf{combine} takes a network and produces the specification given by the colimit of the graph.

\begin{example} \label{ex:networks}
In the example below, the network N3 consists of the nodes S, T2, and U2 and 
two automatically added edges, which are the inclusions from S to T2 and U2.
Thus, N3 is a span, and \textbf{combine N3} yields its pushout. Indeed,
both occurrences of sort s from S are identified in the pushout.

In the network N4, we exclude one of the automatically added inclusions.
Thus, N4 is a graph with one isolated node for T2 and one inclusion edge from S to U2.
\textbf{combine N4} yields the disjoint union of T2 and U2.
That means that the two occurrences of sort s from S are kept seperate.

\begin{lstlisting}
spec S = sort s end
spec T2 = S then sort t end
spec U2 = S then sort u end
network N3 = S, T2, U2 end
network N4 = N3 excluding S -> T2 end
\end{lstlisting}
\end{example}

\section{Desirable Properties of Colimit Selections}\label{sec:properties}
The central definition regarding colimit selection is the following:

\begin{definition}
Given a category $\C$, a \emph{selection of colimits} is a partial
function $\sel$ from $\C$-diagrams $D$ to cocones on $D$ such
that $\sel(D)$, if defined, is a colimit for $D$.
If $\sel$ is only defined for pushouts, we speak of a selection
of pushouts, and so on.
\end{definition}

While it is trivial to give \emph{some} selection of colimits (e.g., by using the axiom of choice or by randomly generating names), it turns out that selecting colimits \emph{elegantly} is a non-trivial task.
For example, selecting a colimit may require inventing new names, or there can be multiple conflicting strategies for selecting names.
\cite{rabe:howto:14} conjectures that it is not possible to select pushouts in a way that the selected pushouts are total, coherent, and enjoy natural names.

In this section, we introduce a suite of criteria for colimit selection.
We work with an arbitrary inclusive category $\C$ with symbols.

\subsection{Symbols of a Diagram}

We are interested in selecting a colimit $(C,\mu_i)$ for a diagram $D:\I\to\C$.
In most practically relevant signature categories, the construction of a colimit can be reduced to the construction of the colimit in $\Set$ of the corresponding sets of symbols.
Because the colimit in $\Set$ amounts to taking a quotient of a disjoint union, we introduce the following auxiliary concept:

\begin{definition}[Symbols of a Diagram]\label{def:diasym}
Given a diagram $D:\I\to\C$, we define the set $\Sym(D)$ by
$$\Sym(D):=\biguplus_{i\in|\I|}|D(i)|:=\{(i,x) \;|\; i\in \I, x\in|D(i)|\}.$$

Moreover, we define the preorder $\leq_D$ on $\Sym(D)$ by
$$(i,x)\ \leq_D\ (j,|D(m)|(x))\text{ for any }m:i\to j\in\I.$$
and we define $\sim_D$ to be the least equivalence relation containing $\leq_D$.

Given any colimit $(C,\mu_i)$ of $D$, we embed $\Sym(D)$ into $|C|$ by defining
\[\mu_D(i,x):=|\mu_i|(x)\]
\end{definition}

Intuitively, $\Sym(D)$ contains the symbols of all nodes of $D$.
$\sim_D$ defines which symbols must definitely be identified in the colimit:
\begin{proposition}
$\sim_D$ is a subset of the kernel of $\mu_D$.
\end{proposition}
\begin{proof}
This follows from $\mu$ being a cocone.
\end{proof}
In some categories such as $\Set$, we even have $\sim_D= \ker(\mu_D)$.

In principle, a natural property to desire of the selected colimit is that $|\sel(D)|$ is a quotient of $\Sym(D)$, in particular $|\sel(D)|=\Sym(D)/\sim_D$ if $\sim_D= \ker(\mu_D)$.
However, that is often impractical, e.g., in the typical case where $|\Sigma|$ is intended to be a set of strings that serve as user-friendly names.
In particular, we do not want to see the indices $i\in I$ creep into the symbol names in $|\sel(D)|$.
Therefore, we define:

\begin{definition}[Names and Name-Clashes]\label{def:nameclash}
For every equivalence class $X\in\Sym(D)/\sim_D$, let $\Nam(X)=\{x|(i,x)\in X\}$.

We say that $D$ is \emph{name-clash-free} if the sets $\Nam(X)$ are pairwise disjoint for all $X$.
We say that $D$ is \emph{fully-sharing} if additionally all sets $\Nam(X)$ have size $1$.
\end{definition}

Intuitively, name-clash-freeness means that whenever two nodes use the same symbol $x$, the diagram requires these two symbols to be shared in the colimit.
An example will be given in Example~\ref{ex:colimit-selections} below.
A particularly common special case arises when both nodes import $x$ from the same node.
The following makes that precise:

\begin{proposition}\label{prop:nameclashfree}
Consider a diagram $D:\I\to \C$.
Assume that for all $(i,x),(j,x)\in\Sym(D)$ there are $(k,y)\in\Sym(D)$ and $m:k\to i$ and $n:k\to j$ in $\I$ such that $|m|(y)=|n|(y)=x$.

Then $D$ is name-clash-free.
If additionally all edges in $D$ are inclusions, $D$ is fully-sharing.
\end{proposition}

The value of name-clash-freeness is the following: for the colimit, we can pick symbols that were already present in $D$.
This allows selecting a colimit whose symbol names are inherited from the diagram (and thus already known to the user who requested the colimit).
Moreover, if $D$ is fully-sharing, these representatives are uniquely determined.\footnote{CASL has a mechanism of ``compound identifiers'' that ensures
  name-clash-freeness in multiple instantiations of parametrised specifications,
such as $List[List[Elem]]$, see \cite{CASL-RM}, p.47f. and p.224f.}

%
%

\subsection{Properties of Colimit Selections}\label{sec:prop}

Being thus prepared, we can now define a number of desirable properties that make a particular selection $\sel$ of colimits elegant.

The most obviously desirable property is that we select a colimit whenever we can:

\begin{definition}[Completeness]
$\sel$ is complete if it is defined for every diagram that has a colimit.
\end{definition}

\subsubsection{Choosing Symbols}\label{sec:prop:symbols}

Typically, we cannot simply choose $|\sel(D)|=\Sym(D)/\sim_D$ because the choice of symbols is restricted:
\begin{definition}[Name-Compliance]
Let $\SYM$ be some subcategory of $\Set$.
We call an object $\Sigma$ $\SYM$-compliant if $|\Sigma|\in\SYM$.
A diagram is $\SYM$-compliant if all involved objects are.

$\sel$ \emph{preserves} $\SYM$-compliance if $\sel(D)$ is $\SYM$-compliant whenever $D$ is.
\end{definition}

In practical systems, symbols must be chosen from a fixed set $S$, e.g., the set of alphanumeric strings.
In that case, $\SYM$ contains all sets that are subsets of $S$.
If we want a compliance-preserving colimit selection, we have to pick names from $S$---that can be much more difficult to do canonically than to pick arbitrary symbols.

It is easy to select colimits by picking arbitrary symbols, e.g., by generating a fresh string as the name of any new declaration.
But that is undesirable---it is preferable that the symbols of $\sel(D)$ are inherited from $D$ in the following sense:

\begin{definition}[Natural Names]
$\sel$ has \emph{natural names} if for every name-clash-free diagram $D$, the selected colimit $\sel(D)=(C,\mu_i)$ is such that
\begin{compactitem}
 \item $|C|$ contains exactly one representative $r\in\Nam(X)$ for every equivalence class $X$,
 \item $|\mu_i|$ maps every $x$ to the respective representative $r$.
\end{compactitem}
\end{definition}

Note that if $D$ is fully-sharing, natural names fully determine $|C|$.
For the general case, we have to choose some $r$ for each equivalence class.
There are multiple options for making that choice canonical.
For example:

\begin{definition}[Origin-Based Names]
Let $\sel$ have natural names.

$\sel$ has \emph{origin}-based symbol names if for every class $X$ the chosen representative $r$ is such that there is some $i$ such that $(i,r)$ is minimal in $X$ with respect to $\leq_D$.
\end{definition}

\begin{definition}[Majority-Based Names]
Let $\sel$ have natural names.

$\sel$ has \emph{majority-based} symbol names if for every class $X$ the chosen representative $x$ maximizes the cardinality of $\{i\in I|(i,x)\in X\}$.

Accordingly, $\sel$ has \emph{majority-origin-based} symbol names if the above cardinality function is used to choose among multiple minimal elements.
\end{definition}

\begin{example}\label{ex:colimit-selections}
Consider a span $D$ consisting of $A\stackrel{\alpha}{\leftarrow}P\stackrel{\beta}{\to}B$.
We consider multiple situations given by the rows of the following table:
\begin{center}
\begin{tabular}{|l|ccccc|}
\hline
    & $|P|$ & $|A|$ & $|B|$ & $|\alpha|$ & $|\beta|$ \\
\hline
1 & $\{\}$ & $\{x\}$ & $\{x\}$ & & \\
2 & $\{p\}$ & $\{a,a'\}$ & $\{b,b'\}$ & $p\mapsto a$ & $p\mapsto b$ \\
3 & $\{p,p'\}$ & $\{a\}$ & $\{p,p'\}$ &~~ $p\mapsto a, p'\mapsto a$~~ & $p\mapsto p, p'\mapsto p'$\\
4 & $\{elem\}$ & $\{nat,+\}$ & $\{elem, list\}$ & $elem\mapsto nat$ & $elem\mapsto elem$ \\
\hline
\end{tabular}
\end{center}

Depending on the situation, different colimit selections are possible:
\begin{compactenum}
\item The diagram is not name-clash-free, and we cannot inherit names.
\item The diagram is name-clash-free but not fully sharing.
The sets $\Nam(\_)$ are $\{p,a,b\}$, $\{a'\}$, and $\{b'\}$.
Thus, there are three possible colimits that have natural names.
All three satisfy the majority condition.
The origin condition allows uniquely selecting $|\sel(D)|=\{p,a',b'\}$.
\item The only set $\Nam(\_)$ is $\{p,p',a,p,p'\}$ (where we repeat elements to indicate how often they occur in the corresponding equivalence class).
We can have natural names, but neither majority nor origin yield a unique choice.
\item This is a typical case of instantiating a parametric specification (here: lists with a parameter for the type of elements) with an actual parameter (here: the set of natural numbers).
The sets $\Nam(-)$ are $\{elem,elem,nat\}$, $\{list\}$, and $\{+\}$.
We can have natural names, and both origin and majority uniquely yield $|\sel(D)|=\{elem,+,list\}$.
However, neither is elegant: The desired choice would be $\{nat,+,list\}$.
\end{compactenum}
\end{example}

\subsubsection{Pushouts along Inclusions}
\begin{wrapfigure}{r}{2cm}
$\xymatrix{
    P\, \ar@{^{(}->}[r] \ar[d]^{\sigma} & B\\
    A
  }$
\vspace{-3em}
\end{wrapfigure}

Pushouts along inclusions are of particular importance because they provide the semantics of parametrization.
As in Section~\ref{sec:constructs}, $D$ is a diagram as given on the right.

The following property is motivated by the desire that instantiating parameterised specifications should always be defined:
\begin{definition}[Total pushouts]
$\sel$ has \emph{total pushouts} if it is defined for all spans where one arrow is an inclusion.
\end{definition}

Moreover, it is desirable that the instantiation extends $A$ in the same way in which $P$ extends $B$.
The following definitions make this precise:

\begin{definition}[Pushout-stable Inclusions]\label{def:pushoutinclusion}
Let $\sel$ have total pushouts.

$\sel$ has \emph{pushout-stable} inclusions if the pushout selection preserves the inclusion, i.e., $\sel(D)$ is of the form
  $$\xymatrix{
    P\, \ar@{^{(}->}[r] \ar[d]^{\sigma} & B \ar[d]^{\sigma^B}\\
    A\, \ar@{^{(}->}[r] & \sigma(B)
  }$$
\end{definition}

\begin{definition}[Pushout-Stable Names]\label{def:pushoutnames}
Let $\sel$ have pushout-stable inclusions.

$\sel$ has \emph{pushout-stable names} if for every selected pushout
  $$\xymatrix{
    P\, \ar@{^{(}->}[r] \ar[d]^{\sigma} & B \ar[d]^{\sigma^B}\\
    A\, \ar@{^{(}->}[r] & \sigma(B)
  }
\hspace{1cm}
\xymatrix{
      |P|\, \ar@{^{(}->}[r] \ar[d]^{|\sigma|} & |B| \ar[d]^{|\sigma^B|}\\
      |A|\, \ar@{^{(}->}[r] & |\sigma(B)|
    }
$$
we have $|\sigma(B)|\setminus |A|=|B|\setminus |P|$ and $|\sigma^B|$ is the identity on that set.
\end{definition}

The aim of pushout-stable inclusions is that we can have
\begin{compactitem}
 \item (vertically) $\_^B$ as a functor $(P\downarrow\C)\to(B\downarrow\C)$,
 \item (horizontally) $\sigma(\_)$ as functor $(P\downarrow\C)\to(A\downarrow\C)$ mapping extensions of $P$ to extensions of $A$.
\end{compactitem}
However, in general, the functoriality laws only hold up to isomorphism.
Therefore, we want to impose an additional condition, which is adapted from \cite{rabe:howto:14}:

\begin{definition}[Coherent Pushouts]\label{def:coherencepushouts}
Let $\sel$ have pushout-stable inclusions.
Then $\sel$ has \emph{coherent pushouts} if the following coherence conditions hold:
   \begin{enumerate}
     \item $id_P(B) = B$ and $id_P^B = id_B$,
     \item $\sigma(P) = A$ and $\sigma^P = \sigma$,
     \item $(\sigma_1;\sigma_2)(B) = \sigma_2(\sigma_1(B))$ 
            and 
           $(\sigma_1;\sigma_2)^B = \sigma_1^B ; \sigma_2^{\sigma_1(B)}$ and finally
     \item for $P\xhookrightarrow{} B_1\xhookrightarrow{} B_2$,
           $\sigma(B_2) = \sigma^{B_1}(B_2)$ and
           $\sigma^{B_2} = (\sigma^{B_1})^{B_2}$
  \end{enumerate}
\noindent where two conditions refer to the following diagrams
$$
\xymatrix{
P\, \ar@{^{(}->}[r] \ar[d]_{\sigma_1}& B \ar[d]_{\sigma_1^B} \ar[ddr]^{(\sigma_1;\sigma_2)^B}\\
A\,  \ar@{^{(}->}[r] \ar[d]_{\sigma_2}& \sigma_1(B) \ar[d]_{\sigma_2^{\sigma_1(B)}} \\
A'\, \ar@{^{(}->}[r] & \sigma_2(\sigma_1(B)) \ar@{=}[r] & (\sigma_1;\sigma_2)(B) 
}\hspace{-3em}
\xymatrix{
 P\, \ar@{^{(}->}[r] \ar[d]_{\sigma} & B_1\, \ar@{^{(}->}[r] \ar[d]_{\sigma^{B_1}}& 
 B_2\, \ar[d]_{\sigma^{B_2}}
     \ar[dr]^{(\sigma^{B_1})^{B_2}}\\
 A\, \ar@{^{(}->}[r] & \sigma(B_1)\, \ar@{^{(}->}[r]& \sigma(B_2)\ar@{=}[r] & \sigma^{B_1}(B_2)
}
$$
\noindent and ensure that pushouts compose vertically and horizontally.
\end{definition}

\subsubsection{Coherence}
The coherence conditions for pushouts can be generalized to arbitrary diagrams.
The general idea is that if there are multiple ways to construct a colimit step-by-step, then it should not matter in which order the construction proceeds.
Here step-by-step means that we first construct a colimit of a subdiagram of $D$ and then add that colimit to $D$ and construct a colimit of the resulting bigger diagram, and so on.

A formal definition for the general case is rather difficult.
The following special case is adapted from \cite{Borceux94}:
\begin{definition}[Interchange]
$\sel$ has \emph{interchange} if given a name-clash-free diagram $D:\I\times\J\to\C$ (seen as a bifunctor)
involving inclusions only \[\sel_{i\in\I}(\sel_{j\in\J} D(i,j))=\sel_{j\in\J}(\sel_{i\in\I} D(i,j))\]
\end{definition}
With an isomorphism instead of equality, this condition always holds.

To state the coherence condition in full generality, we need a few auxiliary definitions:

\begin{definition}\label{def:colimitnode}
Consider a category $I$ with an object $i$ such that every $I$-object has at most one arrow into $i$.

We write $I\setminus i$ for the subcategory of $I$ formed by removing $i$.
We write $I^{\to i}$ for the subcategory of $I$ formed by removing $i$ and all nodes that have no arrow into $i$.
For a diagram $D:I\to \C$, we write $D\setminus i$ and $D^{\to i}$ for the corresponding restrictions of $D$.

We say that $i$ is a \emph{colimit node} of $D$ if $D(i)$ and the set of all morphisms $D(m)$ for $I$-arrows $m$ into $i$ are a colimit of $D^{\to i}$.
If additionally that colimit is equal to $\sel(D^{\to i})$, we call $i$ a $\sel$-colimit node.
\end{definition}

The intuition behind colimit nodes is that they arise by taking a colimit of a subdiagram and can be ignored when forming a colimit of the entire diagram.
For example, in the two commuting diagrams of Def.~\ref{def:coherencepushouts}, the nodes $\sigma_1(B)$ and $\sigma(B_1)$ are colimit nodes.
They arise as the intermediate results of constructing the pushout in two steps.
In general, they arise when constructing a colimit step-by-step:

\begin{proposition}
Consider a diagram $D:I\to \C$ with a colimit node $i$.
Then $D$ and $D\setminus i$ have the same colimits.
\end{proposition}
\begin{proof}
For every $D$-colimit we obtain a $D\setminus i$-cone by removing the injection from $i$.
Vice versa, every $D\setminus i$-colimit $(C,\mu)$ can be uniquely extended to a $D$-cone with the unique factorization $\mu_i:D(i)\to C$ for the colimit $D(i)$.

In both cases, the colimit properties are shown by diagram chase.
\end{proof}

Now we can define that coherence means that we can indeed ignore colimit nodes when selecting a colimit:

\begin{definition}\label{def:coherence}

$\sel$ is \emph{coherent} for the diagram $D$ if for every $\sel$-colimit node $i$ we have that $\sel(D)$ and $\sel(D\setminus i)$ are equal (apart from the former additionally containing the uniquely determined injection $\mu_i$).
\end{definition}

By iterating the coherence property, we can remove or add $\sel$-colimit nodes from/to a diagram without affecting the selected colimit.


\section{Colimit Selections for Typical Signature Categories}\label{sec:colimits}
\subsection{Sets}

As the simplest possible signature category, we consider the category $\Set$ (with standard inclusions and the identity symbol functor).

We first provide a positive result that gives a large sets of desirable properties that can be realised at once: 
\begin{theorem}\label{prop:set-standard-selection}
$\Set$ has a selection of colimits that has
 completeness, pushout-stable inclusions, total pushouts, interchange, and majority-origin-based names.

Moreover, for name-clash-free diagrams, this selection has natural names, pushout-stable names, coherent pushouts.
\end{theorem}

\begin{proof}(Sketch)
If name-clash-freeness is satisfied and the diagram consists of
inclusions only, just take the union as colimit, which ensures that
interchange holds.

Given a span $B \xhookleftarrow{\iota} P
\stackrel{\sigma}{\rightarrow} A$ with $\sigma$ not an inclusion, let
$\sigma(B):=A\cup(B\setminus A)\cup B'$, where $\kappa:B'\cong(B\cap
A)\setminus P$ such that $B'\cap(A\cup(B\setminus A))=\emptyset$. 
Define
  $$\xymatrix{
    P\, \ar@{^{(}->}[r]^(0.3){\iota} \ar[d]^{\sigma} & B = P\cup(B\setminus P)
    \ar[d]^{\sigma^B=\sigma\cup\theta}\\
    A\, \ar@{^{(}->}[r] & \sigma(B)=A\cup((B\setminus (P\cup A))\cup B')
  }$$
where $\theta:B\setminus P\to(B\setminus (P\cup A))\cup B'$ is given by
$$\theta(x)=\twofullcase{\kappa^{-1}(x)}{x\in(B\cap A)\setminus P}{x}
{x\in B\setminus (P\cup A)}.$$

If $D$ is a name-clash-free diagram not of the above forms,
define its colimit $(C, (\mu_i)_{i\in |I|})$ as follows.
$C$ is defined by selecting from each equivalence class 
$X \in \Sym(D)/\sim_D$ a representative $r(X) \in \Nam(X)$
and for each index $i$, and each $(i,x) \in X$,
we define $\mu_i(x) = r(X)$. Use the majority-origin principle.
If that does not determine a representative, select one of the candidates randomly.

Finally, for an arbitrary non name-clash-free diagram,
select an arbitrary colimit, ensuring completeness.
\end{proof}

\ednote{Reviewer suggests giving an example for the construction. TM: instead,
  I have copied the construction from the proof, leaving out the proof
  why the construction works.}

Second, we provide a negative result that gives a small set of desirable properties that cannot be realised at once:

\begin{theorem}\label{prop:set-no-selection}
$\Set$ does not have a selection that has total pushouts, pushout-stable inclusions and names, and coherent pushouts.
\end{theorem}

Thm.~\ref{prop:set-standard-selection} shows that in $\Set$, we can
realise several criteria for colimit selection
we have defined so far.
\medskip


Regarding the choice of names in Thm.~\ref{prop:set-standard-selection}, we cannot expect to achieve origin-based and majority-based names.
In fact, one can show that pushout-stable inclusions and names contradict the origin and majority properties.
Moreover, it is evident that origin and majority
can contradict each other. Consider e.g. 
$$\xymatrix{
\{a\}\ar[r]\ar[d] & \{b\}\ar[d]\\
\{b\}\ar[r] & \{x\}\\
}$$
Origin would lead to $x=a$, while majority would lead to $x=b$.

Nevertheless, the property of origin-majority-based names is useful to guide the pushout selection in cases where the other properties do not determine names uniquely.

\subsection{Product Categories}
  
  Signatures of many logical systems of practical interest
  are often tuples of sets of symbols of different kind. For example, OWL signatures
  consist of sets of atomic classes, individuals, object and data properties.
  To be able to transfer the selection of colimits and its properties defined for
  $\Set$ to categories of tuples of sets, we make use of a more general
  result that ensures that the selection of colimits and its properties are
  stable under products.

  \begin{theorem} \label{thm:products}
    Let $(\C_j)_{j\in J}$ be a family of inclusive
    categories with symbols and assume selections of colimits $sel_j$ 
    that have the properties in Thm.~\ref{prop:set-standard-selection} or Thm.~\ref{prop:set-standard-selection}.
    Then the product $\Pi_{j\in J}C_j$ can be canonically turned into an inclusive
    category with symbols that also has a selection of colimits $\sel$ with the 
    same properties.
  \end{theorem}
  

\begin{example}
  In the case of multi-sorted logics with function or predicate symbols, 
  we can define a selection function for colimits in a step-wise manner.
  Let us consider the case of multi-sorted equational logic, that we denote $EQL$. If we fix a set of sorts $S$, let $\Sign^{EQL}_S$ be the category of
  multi-sorted algebraic signatures with sort set $S$. We can express it as
  \[\Sign^{EQL}_S=\Pi_{w\in S^*, s\in S}\ \Set.\]
  Objects of this category provide a set of operation symbols
  $F_{w,s}$ for each string of argument sorts $w$ and result sort
  $s$.
  With the canonical lifting of the symbol functors of the factors (all of which are the identity on $\Set$) to this
  product, we obtain the symbol functor on $\Sign^{EQL}_S$ given by $|\_|=\biguplus_{i\in J}|\pi_j(\_)|$,
  which decorates each operation symbol with argument and result sorts.
  We write $f:w\to s\in |F|$ instead of $((w,s),f)\in |F|$.
\end{example}


\subsection{Split Fibrations}\label{sec:split}


  Thm.~\ref{thm:products} gives us a selection of colimits for $\Sign^{EQL}_S$. 
  However, our overall goal is to provide such a selection for $\Sign^{EQL}$.
  Now $\Sign^{EQL}$ is a split fibration $\Sign^{EQL}\to\Set$, with fibres
  $\Sign^{EQL}_S$. It is well-known that a split fibration can be obtained
  as Grothendieck construction (flattening) of an indexed category
  indexing the fibres.
  Hence, we will construct such an indexed category for $EQL.$
  This is achieved by observing that each function $u:S\to S'$ leads to a functor
  $B_u:\Sign^{EQL}_{S'}\to\Sign^{EQL}_S$ defined as $B_u(F')=F$, where
  $F_{w,s}=F'_{u(w),u(s)}$. 
  This functor has a left adjoint denoted 
  $L_u:\Sign^{EQL}_{S}\to\Sign^{EQL}_{S'}$
  defined as $L_u(F) = F'$, where $F'_{w',s'} = 
  \uplus_{w\in S^*,s\in S, u(w)=w', u(s) = s'}F_{w,s}$.   
  
  We thus obtain an indexed inclusive category $B:\Set^{op} \to \ICat$,
  and it suffices to show that the selection of colimits and its
  properties are stable under the Grothendieck construction
  (flattening, see \cite{DBLP:journals/tcs/TarleckiBG91}).

  \begin{theorem}\label{thm:Grothendieck}
    Let $B:\Ind^{op}\to \ICat$ be an indexed inclusive category (where $\Ind$ is inclusive
    itself) such that 
     \begin{itemize}
       \item $B$ is \emph{locally reversible}, i.e.~for each $u:i\to j$ in $\Ind$, 
       $B_u:B_j \to B_i$ has a selected left adjoint $F_u: B_i\to B_j$ (note that we do not require coherence of the $F_u$),
        
       \item $\Ind$ has a selection of colimits $\sel_\Ind$,
       \item each category $B_i$ has a selection of colimits $\sel^i$, for $i\in |\Ind|$. 
     \end{itemize}    
    
    Then  
    the Grothendieck category $B^\#$ is itself an inclusive category.\footnote{Note that this construction 
    extends to institutions, yielding Grothendieck institutions, see
    \cite{IIMT}.}
    
  \end{theorem}

\begin{theorem}\label{prop:Grothendieck-symbols}
Under the assumptions of Thm.~\ref{thm:Grothendieck}, let
$(|\_|\theta):B\to\IndSet$ be a (faithful inclusive) oplax indexed
functor (where $\IndSet:\Ind^{op}\to\ICat$ is the
constant functor delivering $\Set$).

This amounts to, for each $B_i$,
a (faithful inclusive) symbol functor $|\_|_i:B_i\to\Set$, and for
each $u:i\to j$, $\theta_u:B_u;|\_|_i\to|\_|_j$ a natural
transformation, such that the $\theta_u$ are coherent.

Then $B^\#$ can be
equipped with a symbol functor as well.
\end{theorem}
\begin{proof}
Define $|(i,A_i)|=|i|\uplus|A_i|_i$, and $|(u:i\to j,\sigma)|=|u|\uplus(|\sigma|_i;(\theta_u)_{A_j})$.
\end{proof}

  \begin{theorem}\label{thm:Grothendieck-properties}
  Under the assumptions of Thm.~\ref{thm:Grothendieck} and Thm.~\ref{prop:Grothendieck-symbols}, extended by:
     \begin{itemize}
       \item $F_u$ preserves inclusions, and moreover, 
       \item the unit
         and counit of the adjunction are inclusions.
     \end{itemize}    
  If $\Ind$ and each $B_i$ have colimit selections enjoying the properties of
  Thm.~\ref{prop:set-standard-selection}, then so does $B^\#$.
  \end{theorem}

  We can apply Thm.~\ref{thm:Grothendieck-properties} to $B:\Set^{op}\to\ICat$ as
  defined above to obtain a selection of colimits $\sel^{EQL}$ for
  $EQL$ signatures. By the theorem,
  $\sel^{EQL}$ has the properties in
  Thm.~\ref{prop:set-standard-selection}.

\begin{example}
We apply these result to $EQL$, where $B_S=\Sign^{EQL}_S$, using the symbol
functors $|\_|_S:\Sign^{EQL}_S\to\Set$ ($S\in|\Set|$) defined above.
Given $u:S\to S'$, $\theta_u:B_u;|\_|_S\to|\_|_{S'}$ is defined as
$(\theta_u)_{F'}:|B_u(F')|\to|F'|$, acting as
$(\theta_u)_{F'}(f:u(w)\to u(s))=f:w\to s$. Using
Thm.~\ref{prop:Grothendieck-symbols}, we obtain the usual symbol
functor for many-sorted signatures, which for any signature delivers
the set of sorts plus the set of typed function symbols of form
$f:w\to s$.

Again, the symbol selection principles of Thm.~\ref{prop:set-standard-selection} carry over.
\end{example}

\section{Categories for Improved Colimit Selection}\label{sec:categories}
\subsection{Named Specifications}

An important technique for avoiding name clashes is to use bipartite IRIs as symbols.
IRIs are Internationalized Resource Identifiers for identification per \nisref{IETF/RFC 3987:2005}.
Symbols using bipartite IRIs consist of
\begin{compactitem}
\item \textbf{namespace}: an IRI that identifies the containing specification, 
usually ending with \syntax{/}\footnote{In some languages, \syntax{\#} is used instead of $/$. But this has the disadvantage that, when used as an IRL, the fragment following the \syntax{\#} is not transmitted to servers.}
\item \textbf{local name}: a name (not containing \syntax{/}) that identifies a non-logical symbol within a specification.
\end{compactitem}
Let $\IRI$ be the subcategory of $\Set$ containing only the sets of bipartite IRIs.


For most practical purposes, it is acceptable to restrict attention to $\IRI$-compliant signatures.
For example, DOL (in accordance with many other languages) strongly recommends using bipartite IRIs.

Note that in an $\IRI$-compliant signature $\Sigma$, the symbols in $|\Sigma|$ may have different namespaces.
For example, in DOL, namespaces $M$ serve as the identifiers of basic specifications $\Sigma$, and then symbols in $|\Sigma|$ are of the form $M/sym$.
But when a specification $N$ imports $M$, (see Sect.~\ref{sec:networks}), the namespace $M$ of the imported symbols is retained and only new symbols declared in $N$ use the namespace $N$.

The main advantage of using IRIs is that specifications (and thus the symbols in them) have globally unique names \cite{W3C:NOTE-swbp-vocab-pub-20080828}.
That makes name clashes much less common:
\begin{proposition}\label{prop:sharing}
Consider a set of basic signatures with pairwise different namespaces.
Then diagrams generated by networks consisting only of $\IRI$-compliant basic specifications and imports are fully-sharing.
\end{proposition}
\begin{proof}
Because basic specifications have unique identifiers, the result follows immediately from Prop.~\ref{prop:nameclashfree}.
\end{proof}
In practice, the assumptions of Prop.~\ref{prop:sharing} quite often
hold, because
networks to be combined often consist of import links only.


\begin{proposition}\label{prop:linkeddata}
Consider $\Set$ with standard inclusions and the identity symbol functor.
The selection constructed in Thm.~\ref{prop:set-standard-selection} can be modified
to a selection that additionally preserves $\IRI$-compliance.
\end{proposition}
\begin{proof}
We just need to ensure that new symbols in the colimit are of the form $N/sym$ for some fresh namespace $N$.
\end{proof}

However, generating fresh namespaces interacts poorly with coherence.

\subsection{Structured Symbol Names}

There are essentially two problems when trying to select colimits canonically: name clashes and ambiguous names.
Intuitively, name clashes arise if we have one name for multiple symbols.
And ambiguity arises if we have multiple names for one symbol.
If neither is the case, named specifications are usually sufficient to obtain canonical colimits.

Our goal now is to handle name clashes and ambiguity.
We introduce a subcategory $\Vocab$ of $\Set$ and focus on $\Vocab$-compliance-preserving colimits.
We want to pick $\Vocab$ in such a way that we can select canonical colimits elegantly.

To motivate the following definition of $\Vocab$, let us look again at the causes behind name clashes and ambiguity.
Name clashes arise if the same node name occurs multiple times in a diagram.
For example, consider two nodes $i$ and $j$ (without any arrows) and $|D(i)|=\{a\}$ and $|D(j)|=\{a\}$.
(This occurs, for example, when taking the disjoint union of the set $\{a\}$ with itself.)
Because this diagram is not name-clash-free, we cannot have natural names in the colimit.
Our solution below introduces qualifiers that create two copies $p/a$ and $q/a$ of the clashing name $a$.

Ambiguity arises if a diagram contains a non-inclusion arrow.
For example, consider $m:i\to j$, and $|D(i)|=\{a\}$ and $|D(j)|=\{b\}$ and $|D(m)|(a)=b$.
$\sim_D$ has one equivalence class, which contains $(i,a)$ and $(j,b)$.
In Section~\ref{sec:prop:symbols}, we focused on choosing either $a$ or $b$ as a natural name in the colimit.
Our solution here retains both names and chooses the set $\{a,b\}$ as a symbol in the colimit.

Because colimits can be iterated, $\Vocab$ must allow for any combination of those two constructions.
That yields the following definition:

\begin{definition}[Structured Symbols]\label{def:structured-symbol}
We assume a fixed set $\Name$ of strings (which we call \textbf{names}).

We write $\QualName$ for the set of lists of names (which we call \textbf{qualified names}).
We assume $Name\subset QualName$, and we write $\nil$ for the empty list and $p/q$ for the concatenation of lists.

A \textbf{structured symbol} is a set of qualified names.

A \textbf{vocabulary} $V$ is a set of pairwise disjoint structured symbols.
We write $\vocdom{V}$ for $\bigcup_{S\in V}S$, and for every $s\in\vocdom{V}$ we write $\vocclass{s}$ for the unique $S\in V$ such that $s\in S$.

We write $\Vocab$ for the full subcategory of $\Set$ containing only the vocabularies.
\end{definition}

The operation $\vocclass{s}$ is crucial: It allows us to use any $s\in S\in V$ as a representative for $S$.
Thus, in order to use structured symbols, we do not have to change our external (human-facing) syntax: Users can still write and read $s$.
We only have to change our internal (machine-facing) syntax by maintaining the set $S$.

$\Vocab$ is an inclusive category with symbols (using the same symbol functor as for $\Set$).
As we see below, it allows for good colimit selections.
But the symbols used in the symbol functor cannot be strings anymore: they are sets of lists of strings.

Above we left open the question where the qualifiers come from that we use to disambiguate name clashes.
It would not be acceptable to use the indices from $I$ as qualifiers because they are arbitrary and not visible to the user.
Instead, we assume that the user has provided qualifiers by assigning labels to some nodes in the diagram:

\begin{definition}[Labeled Diagram]\label{def:labeled-diagram}
A labeled $\C$-diagram $(D,L)$ consists of a diagram $D:\I\to\C$ and a function $L$ from $\I$-objects to $\Name\cup\{\nil\}$.
\end{definition}

$L$ can be a partial function because we only need to label those nodes that are involved in name clashes.
However, it is more convenient to make $L$ a total function by assuming that all unlabeled nodes are labeled with the empty list $\nil$.

Similar to Def.~\ref{def:diasym}, we define the symbols of a labeled diagram:

\begin{definition}[Symbols of a Labeled Diagram]\label{def:labdiasym}
Let $(D:\I\to\Vocab,L)$ be a labeled diagram over $\Vocab$.
We define:
\[\Sym(D,L)=\{(i,L(i)/x) \;|\; i\in \I, x\in\vocdom{D(i)}\}\]
\[(i,L(i)/x)\leq_{DL}(j,L(j)/y) \text{ if for some }m:i\to j\in\I, D(m)(S)=T, x\in S, y\in T\]
$\sim_{DL}$ is the equivalence relation on $\Sym(D,L)$ generated by $\leq_{DL}$.

For every $X\in \Sym(D,L)/\sim_{DL}$, let $\Nam(X)=\{q \;|\; (i,q)\in X\}$.
We say that $(D,L)$ is \emph{name-clash-free} if the sets $\Nam(X)$ are pairwise disjoint.
\end{definition}

Every plain diagram can be seen as a labeled diagram by using $L(i)=\nil$ for all $i$.
In that case, the definition of name-clash-free of Def.~\ref{def:labdiasym} coincides with the one from Def.~\ref{def:nameclash}.

We can now see the power of structured symbols by giving a selection of colimits in $\Vocab$:

\begin{theorem}[Colimits of Vocabularies]
Let $(D,L)$ be a name-clash-free labeled diagram.
Then
\begin{compactitem}
 \item the set $\sel(D,L)$ defined by $\{\Nam(X)\;|\;X\in\Sym(D,L)/\sim_{DL}\}$ is a vocabulary,
 \item the maps $\mu_i:D(i)\to \sel(D,L)$ defined by $\mu_i(\vocclass{x})=\vocclass{L(i)/x}$ are well-defined.
\end{compactitem}
Then $\sel(D,L)$ and the $\mu_i$ form a colimit of $D$.
\end{theorem}

$\sel$ does not exactly have the desirable properties described in Section~\ref{sec:prop}.
But it has variants of them, which is why we recommend $\sel$ as a good trade-off:
\begin{compactitem}
\item $\sel$ is complete in the sense that labels can be added to any diagram to obtain name-clash-freeness.
\item $\sel$ reduces to union for name-clash-free unlabeled diagrams of inclusions (and therefore satisfies interchange).
\item $\sel$ has pushout-stable inclusions for name-clash-free unlabeled diagrams $A\stackrel{\alpha}{\leftarrow}P\hookrightarrow B$ in the sense that all qualified names of $A$ are mapped to themselves in the selected pushout.
\item $\sel$ has natural names in the sense that $L(i)/x$ can be used to identify the corresponding symbol in the colimit, and every symbol in the colimit is of that form.
\item $\sel$ is coherent for all labeled diagrams in which all $\sel$-colimit nodes are unlabeled.
\end{compactitem}
\ednote{FR: The running example about EQL should be picked up here again. But given the space restrictions, it might have to wait until we write a longer version. TM: indeed, this has to wait.}
\ednote{Reviewer asks: Why does the selection sel have these properties? How beneficial are they in practice? What are the implications when implementing colimits, e.g. in Hets?}
\ednote{Reviewer: difficult to assess to what extent this is an improvement over the previous selection techniques --- indeed, further study is needed. TM: have added a sentence in the conclusion.}

\section{Conclusion}\label{sec:conclusion}
We have provided some useful principles for colimit selection, and
studied how far these principles can be actually realised. Some
principles contradict each other, so they need to be prioritised. The
overall goal is to give the user as much control and predictability
over names as possible. This is particularly important for languages
such as CASL and DOL, providing powerful constructs for both
parameterisation and combination of networks, realised through
colimits.  We have shown that our results are stable under products
and Grothendieck constructions; hence they carry over to more complex
signature categories like those of many-sorted logics, HasCASL
\cite{SchroderMossakowski08} (without subsorts) or even categories of
heterogeneous specification (which usually are also obtained via a
Grothendieck construction).

While we have worked with $\Set$ and $\Set$-like categories, future work
should extend the results to more complex categories. E.g.\ the
signature category of the subsorted CASL logic cannot be obtained from $\Set$ by
products and indexing; instead some quotient construction is needed
\cite{CASLColimits}. Another open question is whether coherence
for pushouts can usefully be generalised to other types of colimit.
Moreover, it also will be useful to investigate further 
the pros anc cons of the different selection techniques
(exploitation of name-clash-freeness versus labeled diagrams and structured
symbols) that we have discussed.

One important motivation for this work has been the need to obtain a
better theory for the implementation of colimits in Hets. Currently,
the implementation follows the majority principle only, which led to
complaints from the user community, especially from the Coinvent
project using colimits for conceptual blending. In the future, this
will be revised according to the results of this paper.

\bibliographystyle{plain}
\bibliography{paper}

\newpage
\appendix

\section{Omitted Proofs}
\textbf{Theorem~\ref{prop:set-standard-selection}.}
$\Set$ (with standard inclusions and the identity symbol functor) has a selection of colimits that has
 completeness, pushout-stable inclusions, total pushouts and interchange. 
Moreover, for name-clash-free diagrams, this selection has natural names, coherent pushouts and pushout-stable names.
\begin{proof}
If name-clash-freeness is satisfied and the diagram consists of
inclusions only, just take the union as colimit,
i.e.\ $C=\bigcup_{i\in|I|}D(i)$. This shows that interchange holds.


Given a span $B \xhookleftarrow{\iota} P
\stackrel{\sigma}{\rightarrow} A$ with $\sigma$ not an inclusion, let
$\sigma(B):=A\cup(B\setminus A)\cup B'$, where $\kappa:B'\cong(B\cap
A)\setminus P$ such that $B'\cap(A\cup(B\setminus A))=\emptyset$. 
Define
  $$\xymatrix{
    P\, \ar@{^{(}->}[r]^(0.3){\iota} \ar[d]^{\sigma} & B = P\cup(B\setminus P)
    \ar[d]^{\sigma^B=\sigma\cup\theta}\\
    A\, \ar@{^{(}->}[r] & \sigma(B)=A\cup((B\setminus (P\cup A))\cup B')
  }$$
where $\theta:B\setminus P\to(B\setminus (P\cup A))\cup B'$ is given by
$$\theta(x)=\twofullcase{\kappa^{-1}(x)}{x\in(B\cap A)\setminus P}{x}
{x\in B\setminus (P\cup A)}.$$
Suppose there is any cocone $(C,\nu_A:A\to C,\nu_B:B\to C)$. The mediating
morphism $c:\sigma(B)\to C$ is defined as
$$\threefullcase{\nu_A(x)}{x\in A}{\nu_B(x)}{x\in B\setminus A}{\nu_B(\kappa^{-1}(x))}{x\in B'}$$
This shows that inclusions are pushout-stable, and total pushouts exist.
Moreover, if name-clash-freeness holds for the span $B \xhookleftarrow{\iota} P
\stackrel{\sigma}{\rightarrow} A$, $B'=\emptyset$, hence
$\theta$ is the identity and we have natural names for pushouts.
Using the notation of the coherence diagrams, vertical coherence of
pushouts can be shown as follows: by name-clash-freeness, $(B\cap
A)\setminus P=\emptyset$, hence $B'=\emptyset$ and
$\sigma_1(B)=A\cup(B\setminus (P\cup A))$. Similarly,
$\sigma_2(A\cup(B\setminus (P\cup A)))=A'\cup((A\cup(B\setminus (P\cup
A)))\setminus(A\cup A'))=A'\cup(B\setminus(P\cup A\cup A'))$.  On the
other hand, $(\sigma_1;\sigma_2)(B)=A'\cup(B\setminus(P\cup A')$.  But
since $(B\cap A)\setminus P=\emptyset$, $B\setminus(P\cup A\cup
A')=B\setminus(P\cup A')$, hence
$\sigma_2(\sigma_1(B))=(\sigma_1;\sigma_2)(B)$. Concerning horizontal
coherence, $\sigma^{B_1}(B_2)=\sigma(B_1)\cup(B_2\setminus(B_1\cup\sigma_1(B)))=A\cup(B_1\setminus (P\cup A))\cup(B_2\setminus(B_1\cup A\cup(B_1\setminus (P\cup A))))=A\cup(B_1\setminus (P\cup A))\cup(B_2\setminus(B_1\cup A))=A\cup(B_2\setminus (P\cup A))=\sigma(B_2)$.

In order to show naturalness of names,
let $D$ be a name-clash-free diagram and 
define its colimit $(C, (\mu_i)_{i\in |I|})$ as follows.
$C$ is defined by selecting from each equivalence class 
$X \in \Sym(D)/\sim_D$ a representative $r(X) \in \Nam(X)$
and for each index $i$, and each $(i,x) \in X$,
we define $\mu_i(x) = r(X)$. 
For the particular cases of diagrams that appear in the proof already,
namely those consisting of inclusions only and horizontal and vertical 
compositions of spans with one arrow being an inclusion, the choice
of representative is determined by the respective definitions of the colimit
discussed above.
By name-clash-freeness,
it cannot be the case that $r(X_1) = r(X_2)$ for two equivalence
classes $X_1, X_2$. Thus, we have in $C$ one distinct element for each
equivalence class in $\Sym(D)/\sim_D$ and thus
we have indeed selected a colimit for $D$. 

Finally, for an arbitrary non name-clash-free diagram,
select an arbitrary colimit, ensuring completeness.

The selection constructed above does not satisfy the origin and majority principles.
For example, consider a span $B \xhookleftarrow{\iota} P
\stackrel{\sigma}{\rightarrow} A$ with $\sigma$ not an inclusion.
Then in Thm.~\ref{prop:set-standard-selection},
$\sigma(B):=A\cup(B\setminus A)\cup B'$. This means any symbol from
$P$ that is renamed by $\sigma$ will not appear in the pushout object
$\sigma(B)$, contradicting the origin principle.  Moreover, because
$P\subseteq B$, such a symbol will occur twice (with different
objects) in its equivalence class, but the equivalent symbol from $A$
(occurring only once) is selected in the pushout.

But we can modify the construction to have majority-origin natural names:
For the ``other colimits'' of name-clash-free diagrams, use the majority-origin principle.
If that does not determine a representative, select one of the candidates randomly.
\end{proof}

\textbf{Theorem~\ref{prop:set-no-selection}.}
$\Set$ does not have a selection that has total pushouts, pushout-stable inclusions and names, and coherent pushouts.
\begin{proof}
Assume that such a selection $\sel$ of pushouts were given.
Consider the sequence of two selected pushouts (existing by total pushouts)
$$\xymatrix{ 
\{a\}\,\ar@{^{(}->}[r]\ar[d]^{a\mapsto b} & 
\{a,b\}\,\ar@{^{(}->}[r]\ar[d]^{a\mapsto b,b\mapsto x} & 
\{a,b,x\}\,\ar[d]^{a\mapsto b,b\mapsto x,x\mapsto y} \\ 
\{b\}\,\ar@{^{(}->}[r] & \{b,x\}\,\ar@{^{(}->}[r] &
\{b,x,y\}\, 
}$$ 
where the
presence of $b$ in the first pushout object and of $b$ and $x$ in the
second pushout objects follows from pushout-stable inclusions.
Moreover, $x$ and $y$ are determined by $\sel$, but from the pushout
property we can infer $b\neq x\neq y$.
 Now consider the selected pushout (again existing by total pushouts)
$$\xymatrix{
\{a\}\,\ar@{^{(}->}[rr]\ar[d]^{a\mapsto b} && \{a,b,x\}\,\ar[d]^{a\mapsto b,b\mapsto z,x\mapsto x} \\
\{b\}\,\ar@{^{(}->}[rr] && \{b,z,x\}\, 
}$$

Note that $x$ is determined by $\sel$ in the previous diagram.
Again, the presence of $b$ in the pushout object follows from
pushout-stable inclusions. 
Furthermore, the presence of $x$ in the pushout object and
the mapping $x\mapsto x$ in the pushout inclusion follows
from the pushout-stable names.
Moreover, $z$ is determined by $sel$, but from the pushout
property we can infer $b\neq z\neq x$.

Altogether, in the first diagram, the rightmost pushout injection
maps $b\mapsto x,x\mapsto y$, while in the second diagram, it maps
$b\mapsto z,x\mapsto x$. Since $z\neq x\neq y$, the maps differ.
Hence, coherence does not hold.
\end{proof}

\textbf{Theorem~\ref{thm:products}.}
    Let $(\C_j)_{j\in J}$ be a family of inclusive
    categories with symbols and assume selections of colimits $sel_j$ 
    that have the properties in Thm.~\ref{prop:set-standard-selection} or Thm.~\ref{prop:set-standard-selection}.
    Then the product $\Pi_{j\in J}C_j$ can be canonically turned into a inclusive
    category with symbols that also has a selection of colimits $\sel$ with the 
    same properties.

  \begin{proof}
   The product becomes an inclusive category by using tuples of inclusions as the inclusion morphisms.
   The symbol functors are lifted to the product by defining
   \[|\_|=\biguplus_{i\in J}|\pi_j(\_)|\]
  
   Let $D:\I \to \Pi_{j\in J}C_j$ be a diagram. We can define a selection of colimits
   by taking $\sel(D)$ to be the component-wise combination of $\sel_j(\pi_j(D))$
   where $\pi_j$ are the projections. It is easy to show that
   $\sel$ has the desired properties whenever the $\sel_j$ have them, 
   based on the fact inclusions and colimits in products of inclusive 
   categories are defined component-wise.
  \end{proof}

\textbf{Theorem~\ref{thm:Grothendieck}.}
    Let $B:\Ind^{op}\to \ICat$ be an indexed inclusive category (where $\Ind$ is inclusive
    itself) such that 
     \begin{itemize}
       \item $B$ is \emph{locally reversible}, i.e.~for each $u:i\to j$ in $\Ind$, 
       $B_u:B_j \to B_i$ has a selected left adjoint $F_u: B_i\to B_j$ (note that we do not require coherence of the $F_u$),
        
       \item $\Ind$ has a selection of colimits $sel_\Ind$,
       \item each category $B_i$ has a selection of colimits $\sel^i$, for $i\in |\Ind|$. 
     \end{itemize}    
    
    Then  
    the Grothendieck category $B^\#$ is itself an inclusive category.\footnote{Note that this construction 
    extends to institutions, yielding Grothendieck institutions, see
    \cite{IIMT}.}

  \begin{proof} 
  
    We prove that $B^\#$ is inclusive. Recall that morphisms in $B^\#$
    have form $(i, A_i)\stackrel{(u,\sigma)}{\longrightarrow} (j,
    A_j)$, where $u:i \to j\in\Ind$ and $\sigma: A_i \to B_u(A_j)\in
    B_i$. Now $(u,\sigma)$ is an inclusion if both $u$ and $\sigma$
    are.  The least element is the pair consisting of the least
    element $\emptyset_\Ind$ of $\Ind$ and the least element of
    $B_{\emptyset_\Ind}$.  Given two objects of $B^\#$, $(i, A_i)$ and
    $(j, A_j)$, their union is $(i \cup j, F_{\iota_{i\subseteq i \cup
      j}}(A_i) \cup F_{\iota_{j\subseteq i \cup j}}(A_j))$.  Given a
    class $\{(j,A_j)\}_{j\in J}$ of objects in $B^\#$, their product
    is $(\bigcap J,\bigcap B_{\iota_{\bigcap J\subseteq j}}A_j)$.
    Moreover, we have that a union-intersection square in $B^\#$ is a
    pushout.
  \end{proof}

\textbf{Theorem~\ref{thm:Grothendieck-properties}.}
  Under the assumptions of Thm.~\ref{thm:Grothendieck} and Thm.~\ref{prop:Grothendieck-symbols}, extended by:
     \begin{itemize}
       \item $F_u$ preserves inclusions, and moreover, 
       \item the unit
         and counit of the adjunction are inclusions.
     \end{itemize}    
  If $\Ind$ and each $B_i$ have colimit selections enjoying the properties of
  Thm.~\ref{prop:set-standard-selection}, then so does $B^\#$.

  \begin{proof} 
  From the new assumptions, we can infer that 
\begin{eqnarray}\label{ass1}
\text{$B_u$, $F_u$, $\_^\#$, $\_^\flat$ and $\mathit{Lift}_v$
  preserve inclusions.}\\
\text{The $\tau_{u,v}$ are identities, i.e.\ the $F_u$ are coherent.}\label{ass2}
\end{eqnarray}

We can apply Theorem 2 of
    \cite{DBLP:journals/tcs/TarleckiBG91} to obtain that $B^\#$ has
    colimits. For obtaining selected colimits in $B^\#$, the proof
    from \cite{DBLP:journals/tcs/TarleckiBG91}, which splits colimits
    into coproducts and coequalisers, needs to be replaced by a direct
    proof for all colimits. We recall two preparatory lemmas from
\cite{DBLP:journals/tcs/TarleckiBG91}:
   \begin{lemma} Given index morphisms $u:i\to j$ and $v:j\to k$,
   there is an isomorphism $\tau_{u,v}:F_{u;v}\to F_u;F_v$.
   \end{lemma}
   \begin{lemma}
   Given an index morphism $v:i\to j$, any morphism
   $(u,\sigma):(k,A)\to(i,B)\in B^\#$ can be lifted along $v$
   to a morphism in $B_j$:
   $$\textit{Lift}_v(u,\sigma)=\tau_{u,v};F_v(\sigma^\#):F_{u;v}(A)\to F_v(B)\footnote
   {$\sigma:A\to B_u(B)$, hence its adjoint is $\sigma^\#:F_u(A)\to B$.}$$
   \end{lemma}
   We are now prepared to compute selected colimits in $B^\#$.  Given
   a diagram $D:\I\to B^\#$, let $(c,(\mu_i)_{i\in|\I})$ be the
   selected colimit of $D;\pi_1$, where $\pi_1$ is the projection to
   $\Ind$ (and $\pi_2$ the projection to the second component).
   Define a diagram $D':\I\to B_c$ by $D'(i)=F_{\mu_i}(\pi_2(D(i)))$
   and $D'(m:i\to j)= \mathit{Lift}_{\pi_1(D(j))}(D(m))$.  Let
   $(C,(\nu_i)_{i\in|\I})$ be the selected colimit of $D'$ in
   $B_c$. Then $((c,C),(\mu_i,\nu_i^\flat)_{i\in|\I})$ (where
   $\nu_i^\flat:\pi_2(D(i))\to B_{\mu_i}(C)$ is adjoint to
   $\nu_i:F_{\mu_i}(\pi_2(D(i)))\to C$) is the selected colimit in
   $B^\#$. 

We now prove the properties of the colimit selection:
  \begin{description}
  \item[Natural names]: 
  follows immediately from the assumption   
  that the diagram $D:\I \to B^\#$ is name-clash-free and
  the construction of the colimit in $B^\#$.
    \item[Pushout-stable inclusions:]
    Using the component-wise construction of colimits and (\ref{ass1}).
    \item[Coherent pushouts:] we treat the vertical case only.
Consider the sequence of pushouts
$$
\xymatrix{
(p,P)\, \ar@{^{(}->}[r] \ar[d]_{(u_1,\sigma_1)}& (b,B) \ar[d] \\
(a,A)\,  \ar@{^{(}->}[r]^(0.6){(\iota,\_)} \ar[d]_{(u_2,\sigma_2)}& \bullet \ar[d] \\
(a',A)'\, \ar@{^{(}->}[r]^(0.6){(\kappa,\_)} & \bullet 
}$$
At the index level, we have
$$
\xymatrix{
p\, \ar@{^{(}->}[r] \ar[d]_{u_1}& b \ar[d]_{u_1^b} \ar[ddr]^{(u_1;u_2)^b}\\
a\,  \ar@{^{(}->}[r] \ar[d]_{u_2}& u_1(b) \ar[d]_{u_2^{u_1(b)}} \\
a'\, \ar@{^{(}->}[r] & u_2(u_1(b)) \ar@{=}[r] & (u_1;u_2)(b) 
}$$
At the level of the individual fibres, the two consecutive pushouts are
constructed as:
$$
\xymatrix{
F_{u_1;\iota}P\, \ar@{^{(}->}[r] \ar[d]_{\sigma_1}& F_{u_1^b}B \ar[d]_{(\sigma_1^B)^\#} \\
F_{\iota}A\,  \ar@{^{(}->}[r] & \sigma_1(B) 
}$$
$$
\xymatrix{
F_{u_2;\kappa}A\, \ar@{^{(}->}[r] \ar[d]& F_{u_2^{u_1(b)}}\sigma_1(B) \ar[d]_{(\sigma_2^{\sigma_1(B)})^\#} \\
F_\kappa A'\,  \ar@{^{(}->}[r] & \sigma_2(\sigma_1(B))
}$$
When applying $F_{u_2^{u_1(b)}}$ to the first diagram, by
(\ref{ass2}), both diagrams can be pasted together. The left and upper
legs of this composition are identical to those of the outer pushout
diagram for obtaining $((\sigma_1;\sigma_2)^B)^\#$. Then apply
coherence for $B_{u_2(u_1(B))}$.

    \item[Pushout-stable names:] 
Given a span $B \xhookleftarrow{(\iota,\_)} P
\stackrel{(u,\sigma)}{\rightarrow} A$ in $B^\#$,
 natural names for pushouts give us in $\Ind$:
  $$\xymatrix{
    p\, \ar@{^{(}->}[r]_\iota \ar[d]^{u} & b \ar[d]^{u^b} & \,b\setminus p \ar@{_{(}->}[l]^\kappa  \ar@{=}[d]\\
    a\, \ar@{^{(}->}[r]_\lambda & u(b) & \,u(b)\setminus a \ar@{_{(}->}[l]^\xi
  }$$
and in $B_{u(b)}$:
  $$\xymatrix{
    F_{u;\lambda}P\, \ar@{^{(}->}[r] \ar[d]^{\sigma^\#} & F_{u^b}B \ar[d]^{(\sigma^B)^\#} & \,F_{u^b}B\setminus F_{u;\lambda}P \ar@{_{(}->}[l]  \ar@{=}[d]\\
    F_\lambda A\, \ar@{^{(}->}[r] & \sigma(B) & \,\sigma(B)\setminus F_\lambda A \ar@{_{(}->}[l] 
  }$$
From this, we can obtain natural names for pushouts in $B^\#$.
    \item[Total pushouts:] is implied by completeness.
    \item[Interchange:] follows similar to natural names.
    \item[Completeness:] follows from that for $\Ind$ and $B_i$.
  \end{description}
  \end{proof}

\end{document}